\documentclass[11pt]{article}
\usepackage{epsfig,amsfonts,amsthm,amssymb,latexsym,amsmath,a4wide,tikz,array}
\usepackage{verbatim}
\usepackage{dsfont}
\usepackage{rotating}
\usepackage{multirow}
\usepackage[utf8]{inputenc}

\theoremstyle{plain}
\newtheorem{theorem}{Theorem}

\begin{document}
\date{}

\title{Energy in computing systems with speed scaling: optimization and mechanisms design}
\author{Oscar C. Vásquez\footnote{LIP6, Université Pierre et Marie
    Curie, Paris, France and Engineering Industrial Department, University of Santiago of Chile.}}
\maketitle

\begin{abstract}
We study a simple scheduling game for the speed scaling model. Players want their job to complete early, which however generates a big energy consumption. We address the game from the mechanism design side, and by charging the energy usage to the players we seek for a good compromize between quality of service and energy usage.
\end{abstract}

\section{Introduction}
It is common knowledge that we entered now a period of humanity where natural
resources become rare. This situation triggered consciousness in responsible
consumption, and in particular many countries, companies and individuals aim
in minimizing their energy consumption. Minimizing this resources is a
relatively new topic in decision theory, and opened to new problems and
research areas.

Usually minimizing energy consumption is an opposed goal to maximizing quality of service, think for example at transport systems. In this paper we consider the particular case of computing systems with speed scaling. We study a simplified model, with a unique processor, that can run at a variable continuous and unbounded speed. Users submit jobs to this machine, each job has some workload-representing a number of instructions to execute. Then clearly running these jobs at high speed will be benefic for the users, since it leads to small job completion time which is the quality of service experienced by the users. On the other side running at low speed, will be benefic for the machine, which is then more energy efficient.

In the presence of two opposed goals there is the need to combine them into a
single objective value, which we call the \textit{social welfare}. For this work we decided to represent the energy consumption cost and the quality of service as some utility value. The machine is controlled by a single \textit{efficient regulator} who maximizes the total social welfare, deciding on the speed and the order in which jobs are to be scheduled on the machine.

\section{General Model}
Formally, we consider a non-cooperative game with $n$ users and a regulator. The regulator manages the machine where the jobs are executed. Each user has a job $i$ with a workload $w_i$ and a quality of service function $Q_i$. Since the regulator can observe the workload of job $i$, once it completes, it makes sense to assume that he knows $w_i$ in advance.

The regulator states some \textit{cost sharing mechanism} $C$, which is announced publicly. This mechanism defines a cost share function $C_i$ to be charged to user $i$, which is supposed to compensate the energy consumption cost generated by the execution of the workload $w_i$; and specifies the information $\{\hat{I}_i\}$ requested  by the regulator from user $i$ in order to compute an optimum schedule $s$ for a specific objective. Note that the announced information $\hat{I}_i$ and the real information $I_i$ are not necessarily the same.

In this context a schedule is simply defined by two objects: a speed-function $s$, mapping every time point to a non-negative speed, as well as an ordering of the jobs. We denote this order by a permutation $\sigma$ such that $\sigma(i)$ is the rank of job $i$. For convenience, we denote by $\pi$ the inverse of $\sigma$, i.e. $\pi(j)$ is the $j$-th job to scheduled. In the sequel we identify a schedule by the speed function $s$ and an ordering of the jobs.

The schedule defines a completion time $t_i$ for each job $i$ and generates an energy consumption cost defined as $E(s):=\int_0^{+\infty} s^{\alpha}(t) \text{d}t$, and  $\int_0^{t_{\pi(j)}} s(t) \text{d}t = \sum_{k=1}^j w_{\pi(k)}$ for all $j$, for some constant $\alpha$ which is usually assumed to be $2 \leq \alpha\leq 3$.

Finally, the regulator charges a cost share $C_i$ and announces a completion time $\hat t_i$ for each user $i$, which could be different from $t_i$. By normalization we can assume that the welfare function for the cost share $C_i$ is the identity, i.e. the value of cost share in monetary unit is the same value in utility unit.

The objective of each user $i$ is to minimize his welfare
\[W_i:=Q_i(\hat t_i)-C_i(s),\]
whereas the objective of the regulator is to maximize the total welfare.
\[
WT:=\sum_{i=1}^n W_i.
\]
For a given mechanism, these values depend on the information $\hat I$ declared by the users, so we might adopt a function notation for $WT$.

In this model, the strategy for user $i$ is an information $\hat{I_i}$ that he has to announce. We call $\hat{I_i'}$ as an \textit{improvement} for the user $i$ if his welfare $W_i$ for the strategy profile $(\hat{I_{-i}}, \hat{I_i'})$ is strictly larger than for the strategy profile $\hat{I}=(\hat{I_{-i}}, \hat{I_i})$. Here $(\hat{I_{-i}}, \hat{I_i'})$ stands for the strategy profile which is identical with $\hat{I}$, except at index $i$, where the strategy for player $i$ is $\hat I_i'$. Player $i$ is happy in the strategy profile $\hat I$ if there is no possible \textit{improvement} for him. A strategy profile $I$, where every player is happy, is called a \textit{pure Nash equilibrium}.  As we do not mention mixed Nash equilibria in this work, we will omit the term \emph{pure} from now on.

With these definitions in mind, we are facing two related problems:
 \begin{description}
\item[The Centralized optimization problem] where the regulator has an objective to maximize  and the necessary information for computing an optimum is known to him.
\item[The Mechanism design problem] where the regulator has an objective to maximize  and the necessary information for its computing is only known to the users. In this situation, the regulator must design a cost sharing mechanism such that each user will be happy if he announces his private information as required by the regulator.
\end{description}



\subsection{Desirable properties of cost sharing mechanisms}

We have the game theoretical problem of designing a cost sharing mechanism $C$ with the following desired properties.
\begin{enumerate}
\item The mechanism should be \textit{budget-balanced}, meaning that $\sum_i C_i(s) = E(s)$. If not, we want at least a $\beta$-\textit{budget-balanced} mechanism for some constant $\beta$, meaning that $E(s) \leq \sum_i C_i(s) \leq \beta E(s)$.
\item The mechanism $C$ should always admit a Nash equilibrium. In particular, the mechanism will be \textit{strategy proof} if the strategy profile, where every user announces his private information, is a Nash equilibrium.
\item The mechanism $C$ cannot arbitrarily exclude any users; meaning that if user $i$ is willing to pay enough (at least $C_i$) then his job is executed.
\item In the mechanism $C$, every user has the possibility to exclude himself, resulting that his job won't be executed, the user won't be charged and experience an individual welfare of zero.
\end{enumerate}


\subsection{Users}

We study two types of users according, depending on their welfare functions for quality of service, which can either be $Q_i^A$ and $Q_i^B$ as defined below. Both functions are monotone non-increasing in the completion time of job $\hat t_i$, and translating how important it for user $i$ that his job $i$ completes early. Type A users have some constant utility when the job completes before some private deadline and zero utility if it completes late. Type B users have some utility which is linear with the job completion time. So there is a given utility if their job completes at the ideal time zero, and the utility fades linearly with increasing completion time. The formal definitions of these qualitity of service functions are given below in equations \ref{eq:qualy1} and \ref{eq:qualy2} for some constants $U^A_i,d_i, U^B_i, p_i >0$.

\begin{align}
\label{eq:qualy1}
Q_i^A(\hat t_i ) := &\left\{ \begin{array}{ll}
                 U_i^A & \mbox{if } \hat t_i\leq d_i \\
                 0   & \mbox{otherwise } \end{array} \right.
\\
\label{eq:qualy2}
Q_i^B(\hat t_i) := & U_i^B-p_i \hat t_i
\end{align}

The variants of optimization and mechanism problems are detailed in the next sections.

\section{Type A users}

In this section we consider the case where all users are of type A. This means that every user has a quality of service function $Q_i^A$ using a private value $U^A_i$ and private deadline $d_i$.  The regulator requests the private deadlines from the users.

\subsection{Optimization problem}

We consider the centralized optimization problem consisting in the maximization of the total welfare. Suppose  that the regulator knows the deadlines $d$ of all users and the sum of utility values $\sum_i U_i^A$ is sufficiently large, such that the maximum total welfare is when all users participate to the game. By considering a budget balanced mechanism, we have:

\begin{theorem} Consider a game with type A users. Assume that total welfare is maximized when all users participate to the game. The problem of maximizing the total welfare is equivalent to the problem of minimizing the consumed energy cost which can be computed in time $O(n\log n)$.  \end{theorem}

Minimizing the consumed energy cost has been widely studied, in a more general setting with release times. The paper \cite{YaoDemmersShenker1995} provided the first polynomial time algorithm which time complexity was improved from $O(n^3)$ to $O(n^2 \log n)$ by \cite{LiOn3200discrete}. This has been  improved to $O(n^2)$ if the release time--deadline intervals are laminar, i.e.\ have a tree structure \cite{Min-energytree06}.

In \cite{YaoDemmersShenker1995} it was observed by an averaging argument that in an optimal schedule $s$, if job $i$ is scheduled at speed $s(t)$ at some time $t$, then in the whole span of the job --- which is $[0,d_i)$ in our case --- the speed must be at least $s(t)$. From this fact follows that in an optimal schedule the speed function $s$ is non-increasing and even piecewise constant. So we can describe the schedule by a sequence of pairs of the form $(v,\ell)$ where $v / \ell$ stands for a speed, $\ell$ for the duration of an interval and $v$ for the workload which can be done in this interval $\ell$. We use this argument in our proof.

\begin{proof} (sketch)
Assume that jobs are ordered according to deadlines, i.e. $d_1 \leq \ldots \leq d_n$. By an exchange argument without loss of generality jobs are scheduled by an optimal schedule in the same order. This means that in our case a schedule is really defined by the speed function $s$.
The optimization problem can be formulated as
\begin{equation}
\label{eq:objfun1}
\max \sum_{i=1}^n W_i=\max \sum_{i=1}^n Q_i^A(s)-E(s)= \max \sum_{i=1}^n U_i^A-E(s)
\end{equation}
s.t
\begin{equation}\label{eq:work}
\int_{0}^{d_i} s(t) \text{d}t \geq \sum_{j=1}^i w_j \quad \forall i.
\end{equation}

Since $\sum_{i=1}^n U_i^A$ is constant, the same schedule minimizes $E(s)$ under condition (\ref{eq:work}).  Thus, the mechanism needs to compute a speed function $s$ minimizing the consumed energy while respecting all deadlines.

We describe now a linear time algorithm to compute the optimal schedule $s$, assuming jobs are already ordered according to deadlines. Let $O_i$ be a stack of pairs in decreasing order of speed, describing the optimal schedule for job set $\{ 1, \ldots, i \}$. We define $O_0$ as the empty stack. Now $O_i$ is obtained from $O_{i - 1}$ in the following manner. First we push on $O_{i - 1}$ the pair $( w_i, d_i -d_{i - 1})$. Then while $O_{i - 1}$ contains at least two pairs, and
the two top pair $( v, \ell )$ and second top pair $(v',
\ell')$ satisfy $v' \ell \leq v \ell'$, we replace them by $(v + v', \ell + \ell' )$.
The invariant is that $\sum_{(v,\ell) \in O_i} \ell = d_i$ and the pair $( v, \ell )$ are in strict decreasing order of $\frac{v}{\ell}$
The proof correctness is omitted.

Thus, our specific problem can be solved in time $O(n \log n)$, the necessary time to sort jobs according to deadlines.
\end{proof}


\subsection{Mechanism design problem}

We adopt the above assumptions and define that the information requested by the regulator are the deadlines $d$ of the jobs. This allows him to optimize total welfare.
Thus, only if the information announced by users are the true private values, the regulator is able to optimize total welfare.

We study the \textit{proportional cost} sharing mechanism. Its definition is quite easy, every user is charged exactly the cost generated by the worklaoad of its job. In this mechanism, the regulator annouces the true job completion times, that is $\hat t_i=t_i$.

\begin{theorem}
The game with type A users and the proportional cost  sharing mechanism is strategy proof. 
\end{theorem}

\begin{proof}
  Being strategy proof means that if each player announces the true individual deadlines, i.e. $\hat d_i = d_i$ then the resulting strategy profile is a Nash equilibrium. Let $C_i^p(s)$ be the cost share for player $i$, as defined by the \textit{proportional cost sharing mechanism} $C^p$. We define $C_{\pi(j)}^p(s):=(t_{\pi(j)}-t_{\pi(j-1)})^{1-\alpha} w_{\pi(j)}^\alpha$, where for convenience we abuse notation and write $t_{\pi(0)}:=0$.

We assume that each user $\pi(j)$ wants to maximize his welfare $W_{\pi(j)}(s)$. The mechanism satisfies $\hat{t}_{\pi(j)}=t_{\pi(j)}\leq \hat{d}_{\pi(j)}$. No player knows the information of the other players. We show that for this situation announcing the true value is a dominant strategy, in the sense that it leads to a pure Nash equilibrium. The idea is similar to second-price sealed-bid auction. Table \ref{tab:1} shows the payoff of user $\pi(j)$ in function his strategy and the regulator's strategy.
\begin{table}[ht!]
\caption{Payoff for different strategies in different cases}
\label{tab:1}
\newsavebox\topalignbox
\begin{center}
\newcolumntype{M}{>{$\vcenter\bgroup\hbox\bgroup}c<{\egroup\egroup$}}
\begin{tabular}{|M|M|M|M|M|}
\hline
\multicolumn{2}{|c|}{}&\multicolumn{3}{c|}{User'strategy}\\
\cline{3-5}
\multicolumn{2}{|c|}{}          &$\hat{d}_{\pi(j)}<d_{\pi(j)}$&$\hat{d}_{\pi(j)}=d_{\pi(j)}$&$\hat{d}_{\pi(j)}>d_{\pi(j)}$\\
\hline
\multirow{1}{1em}{\begin{sideways}Regulator's strategy\end{sideways}}&
\begin{sideways}~$t_{\pi(j)}<\hat{d}_{\pi(j)}$~{}\end{sideways}&$(t_{\pi(j)}-t_{\pi(j-1)})^{1-\alpha} w_{\pi(j)}^\alpha$&
$(t_{\pi(j)}-t_{\pi(j-1)})^{1-\alpha} w_{\pi(j)}^\alpha$ & 
$\begin{array}{c}
(t_{\pi(j)}-t_{\pi(j-1)})^{1-\alpha} w_{\pi(j)}^\alpha\\
\mbox{or }0
\end{array}$
\\
\cline{2-5}
&\begin{sideways}~$t_{\pi(j)}=\hat{d}_{\pi(j)}$~\end{sideways}&$(\hat{d}_{\pi(j)}-t_{\pi(j-1)})^{1-\alpha} w_{\pi(j)}^\alpha$&$(d_{\pi(j)}-t_{\pi(j-1)})^{1-\alpha} w_{\pi(j)}^\alpha$ &
0\\
\hline
\end{tabular}
\end{center}
\end{table}

It is easy to see that the strategy $\hat{d}_{\pi(j)}=d_{\pi(j)}$ is dominant.
\end{proof}

In order to measure the proportional cost sharing mechanism, we see that it is \textit{budget-balanced} and it does not exclude any users by definition. In addition, if we assume that each user $i$ has a non negative welfare when his workload is executed and the cost share is charged, then this mechanism will be \textit{efficient}. This latter fact implies uniqueness of the Nash equilibrium which also optimizes total welfare.


\section{Type B users}

In this section we consider the game with type B users. Every player $i$  has now a quality of service function  $Q_i^B$. In this game the regulator needs to receive from the players the workloads $w$ and deadlines $d$ of
the jobs, so he can maximize total welfare. 

\subsection{Optimization problem}

We assume that the regulator knows the penalties $p$ of the  jobs and the sum of utility values $\sum_i U_i^B$ is sufficiently large, such that the total welfare is maximized when all users participate in the game. We consider a budget balanced mechanism, and  we obtain the following result, which has been independently obtained by Megow and Verschae \cite{megow.verschae:dualSchedules}.

\begin{theorem}
Consider a game with type B users. Assume that the total welfare is maximized when all users participate in the game. The problem of maximizing the total welfare is equivalent to the classical scheduling problem denoted as $1||\sum w_i t_i^\beta$ for
$\beta:=(\alpha-1)/\alpha$.
\end{theorem}

\begin{proof}
  Fix some optimal schedule. The jobs complete in some order $\pi$, and by the $j$-th job we refer to this order.  For any $j>1$ the $j$-th job schedules in the interval $[t_{\pi(j-1)},t_{\pi(j)})$ wheras the first job schedules in the interval $[0,t_{\pi(1)})$. Denote by $\ell_j$ the length of the execution interval of the $j$-th job.  From the averaging argument from \cite{YaoDemmersShenker1995} we know that the jobs are scheduled throughout their interval at equal speed $s_j:=w_{\pi(j)} / \ell_j$. With these notations the completion time of the $j$-th job can be written as $\sum_{i=1}^j \ell_i$.

Thus, the optimization problem can be formulated as:
\begin{equation}
\label{eq:objfun2}
\max \sum_{i=1}^n W_i(s)=\max \sum_{i=1}^n Q_i^B(s)-E(s)=\max \sum_{i=1}^n \left( U_i^B-p_i t_i\right) -E(s).
\end{equation}
s.t.
\begin{equation} \label{eq:work2}
\int_{0}^{t_j} s(t) \text{d}t \geq \sum_{i=1}^j w_i \quad \forall j.
\end{equation}

Again, since $\sum_{i=1}^n U_i^A$ is constant, the same schedule maximizes the following expression under condition (\ref{eq:work2}):
\begin{equation}
\min \sum_{i=1}^n p_i t_i+E(s).
\end{equation}

And, by definition of $s(t)$, we have:
\begin{equation}
\min \sum_{j=1}^n w_{\pi(j)}^{\alpha}\ell_j^{1-\alpha} +
     \sum_{j=1}^n p_{\pi(j)} \sum_{i=1}^j \ell_i.
\end{equation}

The first order condition of this function on the variable $\ell_j$ for an
arbitrary index $j$ states
\[(1-\alpha) w_{\pi(j)}^{\alpha} \ell_j^{-\alpha} + \sum_{k=j}^n p_{\pi(k)}=0,\]
which implies
\[ w^{\alpha}_{\pi(j)} \ell_j^{-\alpha} = \frac{\sum_{k=j}^n p_{\pi(k)}}{\alpha - 1}. \]
Replacing these equalities in the social cost value leads to
\begin{align*}
 &\sum_{j=1}^n \ell_j \left(\frac{\sum_{k=j}^n p_{\pi(k)}}{\alpha - 1} + \sum_{k=j}^n p_{\pi(k)} \right)
\\
 = &\frac{\alpha}{\alpha-1} \sum_{j= 1}^n \ell_j \sum_{k=j}^n p_{\pi(k)}
\end{align*}
The first order condition on $\ell_j$ above gave the equality
\[
     \ell_j = \left(\frac{w_{\pi(j)}^\alpha (\alpha-1)}
                                   {\sum_{k=j}^n  p_{\pi(k)}}
                                 \right)^{\frac 1\alpha}
\]
which permits to simplify as
\begin{align*}
 &\frac{\alpha}{\alpha-1} \sum_{j=1}^n \left( \frac{ w_{\pi(j)}^\alpha (\alpha-1) }
                                   {\sum_{k=j}^n p_{\pi(k)}} \right)^{\frac 1\alpha} \sum_{k=j}^n p_{\pi(k)} \\
=& \alpha (\alpha-1)^{\frac{1-\alpha}\alpha} \sum_{j=1}^n w_{\pi(j)} \left(\sum_{k=j}^n p_{\pi(k)} \right)^{\frac{\alpha-1}\alpha}.
\end{align*}

Now consider the following scheduling problem, which also consists of
$n$ jobs, but this time, the processing time of job $i$ is $p_i$, its priority weight is $w_i$ and its completion denoted by $C_i$. The aim is to schedule these jobs on a
single machine, with unit speed, so to minimize \[
\sum_i w_i C_i^{\frac {\alpha-1}\alpha}.
\]

Clearly every optimal schedule to the above problem corresponds to an schedule maximizing total welfare in the former model, just by reversing the job order.  In conclusion, the optimization problem is polynomial equivalent to a classical scheduling problem denoted as $1||\sum w_i C_i^\beta$ for $\beta:=(\alpha-1)/\alpha$.  
\end{proof}

The complexity of this problem type remains still open. Concerning approximation algorithms, Stiller and Wiese \cite{StillerWiese:10:Increasing-speed} show that the Smith's rule \cite{Smith:56:Various-optimizers} guarantees an approximation factor of $(\sqrt{3} + 1)/2 \approx 1.366$ which is tight when $f$ is part of the problem input, whereas H\"ohn and Jacobs \cite{HohnJacobs:12:On-the-performance-of-Smiths} derive a simple method to compute the tight approximation factor of a Smith-ratio-schedule for any particular monotone increasing convex or concave cost function.


\subsection{Mechanism design problem}

From the previous theorem it follows that in order to maximize total welfare, the game regulator needs to know the individual penalties $p$ of the players. Therefore we aim for a thruthful cost sharing mechanism. To fix the notations, we call X the mechanism proposed in this section, and assume that the game regulator announces completion times $\hat t_i$ which are the true completion times $t_i$ of the schedule produced by the regulator.

We define now the cost sharing mechanism X, where  $C_i^x(s)$ is the cost share for player $i$. The mechanism computes the schedule that maximizes total welfare.  By the $j$-th job, we refer to the order in which jobs complete therein.  We define the cost share for the $j$-th player $\pi(j)$ as follows:
\[
           C_{\pi(j)}^x(s):=\sum_{k=1}^{j} \left(\alpha s_k^{\alpha} - \hat p_{\pi(j)} \right)\ell_k
\]
where  
\[
s_k^{\alpha}=\frac{\sum_{r=k}^n \hat p_r}{\alpha-1}=\frac{w_k^{\alpha}}{\ell_k^{\alpha}}.
\]

Note that $C_{\pi(j)}^x(s)$ is equivalent to:
\[(\alpha-1)^{-1} \sum_{k=1}^{j} \left(\hat p_j+\alpha \sum_{r=k,r\neq j}^{n} \hat p_r \right)\ell_k\] where \[
\ell_k = \left( \frac{w_k\sum_{r=k}^n \hat p_r}{\alpha-1} \right)^{-1/\alpha}.
\]
Hence, the cost share for the first player in the sequence is the penalty weighted execution time over all other jobs plus its energy cost.

\begin{theorem}
The game with type B players that the cost sharing mechanism X is strategyproof.
\end{theorem}

\begin{proof}
Consider the definition of $C_i^x(s)$. We assume that each player $\pi(j)$ want to maximize his welfare $W_{\pi(j)}(s)$. We show that PNE is obtained as an local minimum of $W_{\pi(j)}(s)$. By definition, he have that the player $i$ has a welfare is
\[U_i^B-p_i*t_i-(\alpha-1)^{-1} \sum_{k=1}^{\sigma(i)} (\widehat{p_i}+\alpha \sum_{r=k,r\neq \sigma(i)}^{n} \widehat{p_r})\ell_k\]

We now analyze when a changing unilateral of strategy $\hat{p_i}$ does't improve the welfare. We consider the first derive of welfare function.
\[\frac{\partial W_i}{\partial \widehat{p_i}}=0\]
Thus, we have:
\begin{eqnarray*}
0&=&p_i\sum_{k=1}^{\sigma(i)} \frac{\delta \ell_k}{\delta \widehat{p_i}}+ \sum_{k=1}^{\sigma(i)} (\alpha\frac{\delta s_k^{\alpha}}{\delta \widehat{p_i}}-1)\ell_k + (\alpha s_k^{\alpha}-\widehat{p_i})\frac{\delta \ell_k}{\delta \widehat{p_i}}\\
&&\sum_{k=1}^{\sigma(i)} (p_i-\widehat{p_i}) \frac{\delta \ell_k}{\delta \widehat{p_i}}+ \sum_{k=1}^{\sigma(i)} \alpha\frac{\delta s_k^{\alpha}}{\delta \widehat{p_i}}\ell_k -\ell_k + \alpha s_k^{\alpha} \frac{\ell_k}{(1-\alpha)\alpha s_k^{\alpha}}\\
&&\sum_{k=1}^{\sigma(i)} (p_i-\widehat{p_i}) \frac{\delta \ell_k}{\delta \widehat{p_i}}+ \sum_{k=1}^{\sigma(i)} \alpha\frac{\delta s_k^{\alpha}}{\delta \widehat{p_i}}\ell_k -\ell_k +\frac{\alpha \ell_k}{(1-\alpha)\alpha}\\
&&\sum_{k=1}^{\sigma(i)} (p_i-\widehat{p_i}) \frac{\delta \ell_k}{\delta \widehat{p_i}}+ \sum_{k=1}^{\sigma(i)} \frac{\alpha}{\alpha-1}\ell_k -\ell_k -\frac{\ell_k}{\alpha-1}\\
&&\sum_{k=1}^{\sigma(i)} (p_i-\widehat{p_i}) \frac{\delta \ell_k}{\delta \widehat{p_i}}+ \sum_{k=1}^{\sigma(i)} \ell_k \frac{\alpha-\alpha+1-1}{\alpha-1}\\
&&\sum_{k=1}^{\sigma(i)} (p_i-\widehat{p_i}) \frac{\delta \ell_k}{\delta \widehat{p_i}}
\end{eqnarray*}

Given that $\frac{\delta \ell_k}{\delta \widehat{p_i}}\neq 0$, then $\widehat{p_i}$
must be $p_i$.
\end{proof}

\bibliographystyle{plain}
\bibliography{Mechanism}

\begin{thebibliography}{1}

\bibitem{HohnJacobs:12:On-the-performance-of-Smiths}
Wiebke H\"ohn and Tobias Jacobs.
\newblock On the performance of {S}mith's rule in single-machine scheduling
  with nonlinear cost.
\newblock In {\em Proceedings of the 10th Latin American Theoretical
  Informatics Symposium (LATIN '12)}, LNCS, 2012.
\newblock To appear.

\bibitem{Min-energytree06}
Minming Li, Becky~Jie Liu, and Frances~F. Yao.
\newblock Min-energy voltage allocation for tree-structured tasks.
\newblock {\em J. Comb. Optim.}, 11(3):305--319, 2006.

\bibitem{LiOn3200discrete}
Minming Li, Andrew~C. Yao, and Frances~F. Yao.
\newblock Discrete and continuous min-energy schedules for variable voltage
  processor.
\newblock {\em Proceedings of the National Academy of Sciences of the United
  States of America}, 103(11):3983--3987, 2006.

\bibitem{megow.verschae:dualSchedules}
Nicole Megow and José Verschae.
\newblock Scheduling on a machine with varying speed: Minimizing cost and
  energy via dual schedules.
\newblock Technical report, arxiv.org, 2012.

\bibitem{Smith:56:Various-optimizers}
W.E. Smith.
\newblock Various optimizers for single-stage production.
\newblock {\em Naval Research Logistics Quarterly}, 3(1-2):59--66, 1956.

\bibitem{StillerWiese:10:Increasing-speed}
S.~Stiller and A.~Wiese.
\newblock Increasing speed scheduling and flow scheduling.
\newblock {\em Algorithms and Computation}, pages 279--290, 2010.

\bibitem{YaoDemmersShenker1995}
F.~Yao, A.~Demers, and S.~Shenker.
\newblock A scheduling model for reduced cpu energy.
\newblock In {\em Proceedings of the 36th Annual Symposium on Foundations of
  Computer Science}, FOCS '95, pages 374--382, Washington, DC, USA, 1995. IEEE
  Computer Society.

\end{thebibliography}
\end{document}